\documentclass[journal, 10pt]{IEEEtran}
\usepackage[utf8]{inputenc}

\usepackage{amsmath,amssymb,amsfonts,amsthm}
\usepackage{algorithmic}
\usepackage{algorithm}
\usepackage{graphicx}
\usepackage{xcolor}
\usepackage{mathtools} 
\usepackage{tikzpagenodes}
\DeclareMathOperator*{\argmin}{argmin}

\usepackage[style=ieee, backend=biber]{biblatex}
\newtheorem{theorem}{Theorem}
\newtheorem{lemma}{Lemma}
\theoremstyle{definition}

\DeclareMathOperator{\prox}{prox}

\newcommand{\indep}{\perp \!\!\! \perp}

\title{Asynchronous Multiuser Detection for SCMA with Unknown Delays: A Compressed Sensing Approach}

\author{
    Dylan Wheeler, \IEEEmembership{Student Member, IEEE,} Erin E. Tripp, and Balasubramaniam Natarajan, \IEEEmembership{Senior Member, IEEE}

    \thanks{This work was funded in part by AFOSR grant 21RICO035. Any opinions, findings and conclusions or recommendations expressed in this material are those of the authors and do not necessarily reflect the views of the U.S. Air Force Research Laboratory. Cleared for public release 11 August 2022: Case number AFRL-2022-3871.}% 
    \thanks{D. Wheeler and B. Natarajan are with the Mike Wiegers Department of Electrical and Computer Engineering at Kansas State University (email: dylan84@ksu.edu)}
    \thanks{E.E. Tripp is with the High Performance Systems Branch of the Air Force Research Laboratory Information Directorate (email: erin.tripp.4@us.af.mil)}
}

\begin{document}

\maketitle
\begin{tikzpicture}[remember picture,overlay]
    \node[align=center] at ([yshift=1em]current page text area.north) {This work has been submitted to the IEEE for possible publication. Copyright may be\\ transferred without notice, after which this version may no longer be accessible.};
\end{tikzpicture}%
\thispagestyle{plain}

\begin{abstract}
    Despite being the subject of a growing body of research, non-orthogonal multiple access has failed to garner sufficient support to be included in modern standards. One of the more promising approaches to non-orthogonal multiple access is sparse code multiple access, which seeks to utilize non-orthogonal, sparse spreading codes to share bandwidth among users more efficiently than traditional orthogonal methods. Nearly all of the studies regarding sparse code multiple access assume synchronization at the receiver, which may not always be a practical assumption. In this work, we aim to bring this promising technology closer to a practical realization by dropping the assumption of synchronization. We therefore propose a compressed sensing-based delay estimation technique developed specifically for an uplink sparse code multiple access system. The proposed technique can be used with nearly all of the numerous decoding algorithms proposed in the existing literature, including the popular message passing approach. Furthermore, we derive a theoretical bound regarding the recovery performance of the proposed technique, and use simulations to demonstrate its viability in a practical uplink system. 
\end{abstract}

\begin{IEEEkeywords}
    SCMA, NOMA, delay estimation, beyond-5G, 6G
\end{IEEEkeywords}

\vspace{-0.4cm}

\section{Introduction}

As fifth generation (5G) wireless network technologies continue to roll out across the globe, the research community has begun to shift their efforts toward the development of the next generation of wireless network, commonly referred to as either beyond-5G or 6G. Some have set forth their visions of what this future network might look like \cite{strinati_6G}. One prominent idea involves the implementation of non-orthogonal multiple access (NOMA) techniques. NOMA is seen as a key technique to achieve higher spectral efficiency in some scenarios, e.g. massive multiple access \cite{scma6G} that can enable greater efficiency in massive internet of things (IoT) applications. Although it was first seen as a promising technique for 5G, NOMA has failed to garner sufficient support to be included in the current 5G standards \cite{3gpp_rel16}. If NOMA is to play a larger role in the wireless network of the future, further improvements must be made.

Two prominent NOMA techniques include power domain NOMA and code domain NOMA \cite{noma}. In power domain NOMA, different transmit power levels are used to distinguish among different users' data at the receiver using a technique known as successive interference cancellation (SIC). In essence, the data corresponding to the greatest power level is decoded first, with all other data treated as noise. This data is then removed from the observation, with the decoding of the subsequent data following in the same manner. In code domain NOMA, non-orthogonal codes are used to encode and spread data symbols over some orthogonal resource elements (RE's), such as orthogonal frequency tones or time slots. Since non-orthogonal codes are used, collisions of data symbols may occur on a given RE, resulting in increased difficulty when decoding at the receiver. Despite the complications introduced by these methods, both have been shown to lead to more efficient use of increasingly precious spectral resources \cite{noma}.

One particular method of code domain NOMA that has drawn much interest over the past decade is sparse code multiple access (SCMA). First introduced by Nikopour and Baligh \cite{scma_nikopour}, SCMA limits the total interference on any orthogonal RE by using sparse, non-orthogonal codes to spread the data. Considering the uplink scenario, this is accomplished by mapping each user's bits directly to sparse, complex-valued codewords from a codebook unique to each user. The sparsity pattern of a user's codebook determines which RE's will be used by that user, and dictates the total interference on any given RE. With knowledge of the individual codebooks at the receiver, simultaneous decoding of each user's data, i.e. multiuser detection (MUD), is possible through an iterative algorithm known as the message passing algorithm (MPA) \cite{scma_nikopour}. This algorithm takes in the signals received at each RE as inputs, and returns log-likelihood ratios which are then used to directly make decisions on each user's bits.

Over the past decade, numerous variants of MPA and other low complexity alternatives for SCMA receiver design have been proposed. However, these efforts typically assume synchronous reception. With an expected explosion of IoT devices in the post-5G era, the synchronous assumption may not be realistic and could limit the applicability of SCMA. Therefore, in this paper we introduce a delay estimation scheme designed for an SCMA uplink system. The proposed scheme makes use of modern compressed sensing techniques for efficient and reliable delay estimation, and can be implemented alongside the vast majority of SCMA decoders that have been proposed in the existing literature.

\subsection{Related Work}
\label{rel_work}

MPA was the first proposed method for decoding in an SCMA system \cite{scma_nikopour}. Although MPA can achieve near-optimal performance, it is still computationally prohibitive for practical implementation. Thus, many works have focused on developing SCMA decoding schemes with decreased complexity. A recent and thorough survey discussing these methods, among other aspects of SCMA, is provided by Rebhi \textit{et al.} \cite{scma_survey}. Here we briefly discuss some notable schemes that are considered in finer detail in the aforementioned survey.

As MPA involves many exponential calculations, some variations exploit the monotonic property of the logarithm function and are termed Max-Log-MPA and Log-MPA \cite{maxlogmpa}. The performance of these algorithms converges to that of MPA as SNR increases, though the complexity is still prohibitive for practical implementation. Some other MPA-based techniques include the utilization of a lookup table \cite{mpa_dai}, fairness-based calculations \cite{mpa_han}, and partial marginalization \cite{mpa_jia}, among others. Despite the various methods employing this technique, the complexity of MPA always grows exponentially with the degree of non-orthogonality, leading to impractical schemes in the massive scenario.

Another method that has been proposed to reduce the complexity of MPA is sphere decoding. In sphere decoding, the search space of the decoder is reduced to only those codewords in a given hypersphere around the received signal vector. In \cite{scma_tian}, the authors utilize Log-MPA with restricted search regions. Other works have developed detectors which take a joint approach to both SCMA and channel decoding. These decoders can achieve an additional coding gain by performing the decoding in a joint fashion. Two coding schemes have been adopted for use in 5G, namely low-density parity-check (LDPC) codes and polar codes \cite{3gpp_rel16}. In \cite{scma_lai}, the sparse factor graphs of both LDPC coding and SCMA are combined to form a single factor graph over which a simplified version of MPA is carried out. The authors of \cite{scma_mu} propose a simple combination of the MPA detector for SCMA and a soft-input soft-output successive cancellation for polar coding.  

% Yet another approach is to take advantage of artificial intelligence, or more specifically, deep learning. Autoencoders use deep learning to determine optimal coding strategies in a purely data-driven manner. In \cite{scma_kim}, denoising autoencoders are used to design the encoder and decoder of an SCMA system in an end-to-end fashion. Another application of deep learning has been termed algorithm unfolding (or unrolling), and involves replacing each iteration of an model-based algorithm with a deep neural network which can be trained with data. This was the approach taken in \cite{scma_lu}, in which each iteration of MPA is carried out by a deep neural network, improving complexity. Similarly, the authors of \cite{scma_sun} apply this approach to perform joint channel and SCMA decoding.

A different approach that bypasses the exponential complexity issue of MPA includes algorithms which perform distribution approximation. The authors of \cite{scma_huang} develop simplified algorithms by approximating other users' messages as Gaussian interference. Another approximation method utilizes a Markov Chain Monte Carlo (MCMC) method to draw samples in an iterative manner which converge to the most likely combination of codewords \cite{chen}. As these methods, along with deep learning approaches, reduce decoding complexity from exponential to linear, they are strong candidates for practical implementations of SCMA.

While there have been many techniques developed for SCMA decoding, utilizing a diverse set of approaches, nearly all of them study the problem under the assumption of synchronous reception. However, synchronized reception at the receiver is not always practical, especially when there are many devices operating on the system, as in the massive machine type communication (mMTC) use case defined for 5G. Therefore, it is of interest to study SCMA detection scheme under the assumption of asynchronous reception. There is very limited prior work related to this problem. For example, the authors of \cite{async_scma} study the problem of SCMA detection in the presence of perfectly known delays. They propose a sampling scheme which makes use of delay knowledge, and then apply a combined belief propagation and message-passing algorithm, demonstrating similar performance to synchronous MPA. While this work introduces a scheme for asynchronous SCMA detection, perfect delay knowledge is likely not available in a practical setting, and the asynchronous SCMA detection problem under the condition of unknown delays appears to be entirely absent from the current literature.

Therefore, we are interested in the problem of delay estimation in an SCMA system. A related problem is that of delay estimation in a general multiple access system, such as orthogonal frequency division multiple access (OFDMA). One recent work studies the problem of delay estimation in OFDMA uplink channels specified by the DOCSIS 3.1 standard \cite{delay_est_nguyen}, in which an iterative method centering on peak detection of the inverse fast Fourier transform of a received pilot sequence. Another paper utilizes a compressed sensing technique  \cite{delay_est_senyuva} similar to the technique employed in our work. However, their method is specific to an OFDMA system, while our formulation can generally apply to SCMA systems utilizing any orthogonal REs, not just OFDM tones. In \cite{delay_est_masmoudi}, a delay estimation algorithm is proposed for general multicarrier direct sequence code-division multiple access (MC-DS-CDMA) systems, of which SCMA is a particular case, making use of the expectation maximization (EM) algorithm to jointly consider all subcarriers. While this work captures SCMA as a particular case, a delay estimation algorithm designed specifically for the SCMA system appears to also be absent from the literature.

\begin{figure*}[htbp]
    \centerline{\includegraphics[scale = 0.7]{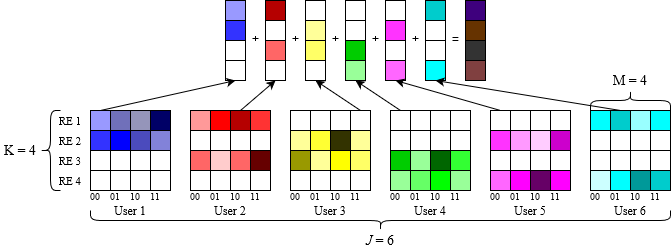}}
    \caption{Illustration of SCMA system with $J=6$ users and $K=4$ orthogonal resource elements. }
    \label{scma_fig}
\end{figure*}

\subsection{Contributions}
To address the limitations of existing approaches for asynchronous SCMA detection, this work aims to 
\begin{itemize}
    \item Derive an estimation scheme for recovering unknown delays at the receiver in an SCMA system. This scheme can be used to support any previously proposed decoders in the asynchronous scenario.
    \item Propose an efficient method, utilizing convex optimization techniques, to perform delay estimation. Demonstrate the ability of this method to successfully recover delay information through both theoretical and simulated results.
    \item Combine efficient delay estimation with previously proposed SCMA codebook design and detection techniques to demonstrate overall performance in terms of bit-error rate (BER) at the SCMA uplink receiver compared to synchronous methods.
\end{itemize}

The rest of this paper is organized as follows. Section \ref{mod_sec} presents the system model that is used throughout the paper. Details about the synchronous and asynchronous model along with the associated assumptions are provided. In section \ref{app_sec}, the proposed asynchronous SCMA delay estimation scheme is provided. Additionally, some results regarding the convergence of the estimation algorithm are summarized, and some theoretical results regarding recovery guarantees are derived. In section \ref{sim_section}, the simulation setup is described and results are given which demonstrate the performance of the proposed technique. Section \ref{con_sec} concludes the paper, and gives some ideas for future research directions.

\section{System Model}
\label{mod_sec}

In this section, we describe in detail the aspects of the model used for the SCMA system. The traditional SCMA model for the synchronous case is presented first, followed by a description of the asynchronous scenario.

\subsection{Synchronous SCMA}

As was briefly mentioned above, in uplink SCMA each user's bits are mapped directly to sparse, complex-valued codewords drawn from a codebook that is unique to each user. In particular, consider a system servicing $J$ total users with $K$ orthogonal REs, where it is assumed that $J > K$, i.e. the system is \textit{overloaded}. Each user $j = 1, \ldots, J$ has their own unique codebook $\mathcal{C}_j$ of cardinality $M$ with unique sparsity pattern, as is illustrated in Figure \ref{scma_fig}. 

Formally, the received SCMA discrete signal can be modelled as
\begin{equation}
    \textbf{y} = \sum_{j=1}^J \textbf{H}^{(j)}\textbf{x}^{(j)} + \textbf{z},
    \label{eq_sync_model}
\end{equation}
where $\textbf{y}$ denotes the received signal vector of dimension $K\times 1$, $\textbf{H}^{(j)}_{K\times K}$ denotes the $K \times K$ diagonal channel matrix corresponding the $j^{th}$ user, $\textbf{x}^{(j)}$ denotes the $K \times 1$ codeword selected by the $j^{th}$ user, and $\textbf{z} \sim \mathcal{CN}(\textbf{0}, \sigma^2 \textbf{I})$ is a $K \times 1$ vector of complex additive white Gaussian noise (AWGN). We denote by $x_k^{(j)}$ the $k^{th}$ element of $\textbf{x}^{(j)}$. 

Referring to Figure \ref{scma_fig}, each colored square represents a nonzero, complex entry of the corresponding codeword (column) that represents some quadrature amplitude modulated (QAM) symbol, while the white squares represent zero entries. Thus, in the setup depicted each codeword is 2-sparse, meaning that each user's data is spread over two RE's. Conversely, each RE is occupied by three users; e.g. RE 1 is utilized by users 1, 2, and 6.

Observe that this formulation describes the system over a single QAM symbol duration, and let this duration be denoted by $T_s$. Thus, the discrete superposition of QAM symbols (with added noise) on a given RE is recovered by sampling over the duration $T_s$. Doing this at every RE results in recovery of the vector $\textbf{y}$, which is then used as an input, along with channel state information, to the decoder to directly make decisions on each user's bits. The same basic process is followed in the case that a sequence of symbols is transmitted by each user; namely, the superimposed symbols are sampled at each RE and decoded in a sequential manner. One important assumption for the success of this method is synchronization in time of the reception of each user's transmission. Without perfect synchronization, simply sampling over $T_s$ will result in \textit{intersymbol interference}, greatly impacting the performance of the system \cite{goldsmith_2005}.

\subsection{Asynchronous SCMA Reception}

To observe the effect of imperfect synchronization, first consider the continuous-time complex baseband representation of the aforementioned discrete signals. Assume that the $N$-symbol data sequences transmitted by each user $j = 1,2,\ldots,J$ experience some delay, denoted by $\tau_j$. These delays could be due to factors such as user mobility and multipath effects, which make synchronization difficult in a massive network. For ease of discussion, we will consider here only the received signal on RE $k$. Furthermore, denote by $\Omega_k$ the set of all users sharing this RE, i.e.,
\begin{equation}
    \Omega_k = \{ j \mid x_k^{(j)} \neq 0\}.
\end{equation}
Then the signal received by RE $k$ is expressed as \cite{async_scma}
\begin{equation}
    y_k(t) = \sum_{n=1}^N \sum_{j \in \Omega_k} h_k^{(j)}x_k^{(j)}[n] r(t - nT_s - \tau_j) + z_k(t),
\end{equation}
where $x_k^{(j)}[n]$ denotes the $n^{th}$ QAM symbol in the $N$-symbol sequence transmitted by user $j$ over RE $k$, $h_k^{(j)}$ denotes the complex channel coefficient corresponding to user $j$ and RE $k$ (assumed to be fixed for the transmission interval), $r(\cdot)$ denotes a pulse-shaping function, and $z_k(t)$ is the time-domain expression for the AWGN for RE $k$. We will assume that $r(\cdot)$ is simply a rectangular pulse function. A visual demonstration of this signal on RE 1 for the system presented in Figure $\ref{scma_fig}$ is given in Figure \ref{async_scma_fig}.

\begin{figure*}[htbp]
    \centerline{\includegraphics[scale = 0.50]{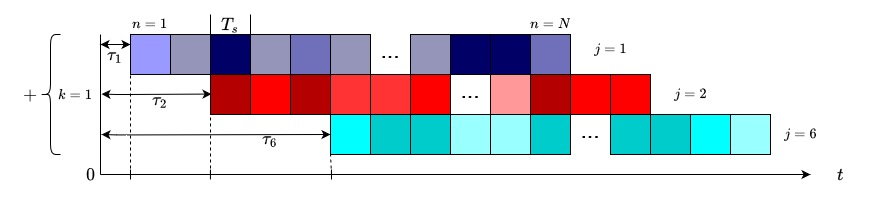}}
    \caption{Visual representation of the received signal at RE 1 in the asynchronous scenario for the system depicted in Figure \ref{scma_fig}}
    \label{async_scma_fig}
\end{figure*}

Suppose that at RE $k$, the receiver synchronizes to the first received signal, e.g. user 1's signal in Figure \ref{async_scma_fig}. If synchronization is achieved, i.e., $\tau_i = \tau_j$ for all $i,j \in \Omega_k$, then taking a sample every $T_s$ seconds will yield the following output
\begin{equation}
    y_k[n] = \sum_{j \in \Omega_k} h_k^{(j)} x_k^{(j)}[n] + z_k[n]
    \label{sync_superpos}
\end{equation}
which corresponds exactly to the model first defined in (\ref{eq_sync_model}). In the case where synchronization is not achieved, by observation of Figure \ref{async_scma_fig} it becomes clear that this approach is no longer applicable.

A simple solution arises when we assume perfect knowledge of the delays for each user is available, as is presented in \cite{async_scma}. The key idea is that, in order to apply MPA (or any other decoding algorithm), we need a discrete received signal in which each symbol can be represented as a summation of \textit{some} input symbols; not necessarily from the same discrete time instance $n$. Consider again the scenario depicted in Figure \ref{async_scma_fig}, and suppose we sample at a rate of $1/T_s$. With knowledge of the delays, we know that the first two samples contain only symbols from user 1. Then, samples 3-5 contain only symbols from users 1 and 2. Continuing in this manner, we are able to classify each sample of the received signal as a summation of some noisy, distorted input symbols. Doing this for every RE, MPA can then be applied on a sample-by-sample basis to perform MUD.

As we do not assume \textit{a priori} knowledge of the delays, we cannot directly implement the scheme described above. Therefore, we assume a sampling rate of $1/T_s$ at each RE, and a few assumptions are held for simplicity. First, we assume that each delay is an integer multiple of the sampling period, or in other words, each element of the received signal on RE $k$ is a perfect superposition of symbols from users sharing that RE, as is illustrated in Figure \ref{async_scma_fig}. Note that by decreasing the symbol period $T_s$, we can achieve any arbitrary level of sampling resolution at the cost of increased bandwidth. Second, it is assumed that all delays are bounded by some maximum delay $\tau$, i.e. $\tau_j < \tau \: \forall \: j = 1,2,\ldots,J$. This is reasonable given that some coarse synchronization method for the system is already in place, e.g. an IoT setup in which a central base station periodically requests data from some sensors. In this case, the maximum delay would correspond to the maximum delay incurred from sensing, processing, and transmission of the data at each node.

\section{Proposed Approach}
\label{app_sec}

As mentioned in the previous section, in order to apply any of the numerous SCMA decoding algorithms, we require knowledge of \textit{which} symbols we are decoding. Let's clarify this with an example. Assume we are using an MPA-based decoder and consider Figure \ref{async_scma_fig}. The decoder will simultaneously act on the signals from all $K$ RE's on a sample-by-sample basis to decode the received signal. For example, the fifth sample of the signal in Figure \ref{async_scma_fig} will be considered with the fifth sample of all other RE's to produce LLR's for the bits corresponding to that sample time. However, note that the fifth sample of the received signal corresponds to the \textit{fifth} symbol of the first user and the \textit{third} symbol of the second user on that RE. Thus, delay knowledge is required to properly associate the LLR's obtained from the MPA decoder to the correct symbols (equivalently, the correct bits) of each user. Operating under the assumption that delay knowledge is unknown \textit{a priori}, we require a method to estimate these delays. In this section we describe our proposed method for delay estimation, as well as some theoretical guarantees based on this method. % Note that our estimation approach is independent of the chosen decoding algorithm, and that any of the methods described in Section \ref{rel_work} could be implemented.

\subsection{Delay Estimation with Compressed Sensing}
\label{main}

Our main contribution is a delay estimation technique inspired by the method employed in \cite{cdra_async}, which uses a convex optimization approach for active user detection in an asynchronous code-division random access setup. Assuming a multiple access scenario instead, we will see that a similar approach can yield efficient and robust delay estimation. First of all, let the vector of transmitted symbols from user $j$ on RE $k$ be structured as
\begin{equation*}
    \textbf{x}_k^{(j)} = \begin{pmatrix}
        \textbf{s}^{(j)}_{1\times S} & \textbf{d}^{(j)}_{1\times N}
    \end{pmatrix}^T,
\end{equation*}
where $\textbf{s}^{(j)}$ is an $S\times 1$ vector of pilot symbols, $\textbf{d}^{(j)}$ is the $N \times 1$ vector of data symbols. Thus, the entire transmitted frame will consist of a total of $S+N$ discrete symbols. Furthermore, define the delayed symbol frame transmitted by user $j$ over RE $k$ as
\begin{equation*}
    \tilde{\textbf{x}}_k^{(j)} = \begin{pmatrix}
        \textbf{0}_{1 \times d_j} & \textbf{x}_k^{(j)} & \textbf{0}_{1 \times D - d_j}
    \end{pmatrix}^T,
\end{equation*}
where $d_j$ is defined as the \textit{discrete} delay of user $j$'s signal, and $D$ is the maximum discrete delay. Working under the assumption that each delay $\tau_j$ is an integer multiple of the sampling period $T_s$, we have that $D = \tau/T_s$, and $\tau_j = d_j T_s$, with $d_j \in \{0, 1, \ldots, D\}$. We see that the dimension of $\tilde{\textbf{x}}_k^{(j)}$ is $(S+N+D) \times 1$. Recall that the set of users sharing RE $k$ is denoted by $\Omega_k$, and define the cardinality of this set as $\eta_k$. Let the $i^{th}$ element of $\Omega_k$ be $j_{ki}$. Then we can define the discrete received signal on RE $k$ in terms of the $\tilde{\textbf{x}}_k^{(j)}$'s as
\begin{equation}
    \textbf{w}_k = \begin{bmatrix}
        \tilde{\textbf{x}}_k^{(j_{k1})} & \tilde{\textbf{x}}_k^{(j_{k2})} & \cdots & \tilde{\textbf{x}}_k^{j_{k\eta_k}}
    \end{bmatrix} \textbf{h}_k + \textbf{z}_k,
    \label{delay_model}
\end{equation}
where
\begin{equation*}
    \textbf{h}_k = \begin{pmatrix}
        h_k^{(j_{k1})} & h_k^{(j_{k2})} & \cdots & h_k^{(j_{k\eta_k})}
    \end{pmatrix}^T.
\end{equation*}
The linear model (\ref{delay_model}) captures the behavior of the asynchronous SCMA system on RE $k$ over a transmission frame as a noisy sum of delayed signals. Our goal is to use this model to estimate the delays $d_j$ for the users sharing RE $k$. However, in (\ref{delay_model}), the unknown delays are represented by the structure of the symbol matrix, namely by the number of zeros that precede the first nonzero symbol in each column. Thus, in its current form, the linear model is composed of an \textit{unknown} matrix multiplied by a \textit{known} vector. We will first reformulate this problem into a form of a \textit{known} matrix multiplied by an \textit{unknown} vector. To achieve this, let the length of each pilot sequence be greater than $D$, and suppose that the last $D$ symbols of each pilot sequence take a value of 0. Then we can observe that the first $S$ symbols of each $\tilde{\textbf{x}}_k^{(j)}$ will contain only nonzero pilot symbols and zeros. Furthermore, this vector will contain the entire nonzero pilot sequence, due to the assumption on the maximum delay. Denote this nonzero sequence by $\textbf{s}'^{(j)}$. Working with the known sequence  $\textbf{s}'^{(j)}$, we can form the matrix
\begin{equation}
    \textbf{T}^{(j)} = \begin{bmatrix}
        \textbf{s}'^{(j)} & 0 & \cdots & 0\\
        0 & \textbf{s}'^{(j)} & & \vdots  \\
        \vdots &  & \ddots & 0 \\
        0 & \cdots & 0 & \textbf{s}'^{(j)}\\
    \end{bmatrix}.
\end{equation}

The first column is simply $\textbf{s}^{(j)}$. In the second column, $\textbf{s}'^{(j)}$ is ``shifted'' down one; i.e., a zero is added above and removed below $\textbf{s}'^{(j)}$. This continues until there are $D$ zeros above $\textbf{s}'^{(j)}$ and none below, and this is the last column. The dimensions of $\textbf{T}^{(j)}$ are then $S \times (D+1)$.

Note that each column represents a different delay possibility for user $j$. Furthermore, under the assumptions on the delays stated above, this matrix lists \textit{all} possible discrete delays for this user. Repeating this process for the other users sharing RE $k$, we can define the block matrix
\begin{equation}
    \label{eq_Tk}
    \textbf{T}_k = \begin{bmatrix}
        \textbf{T}^{(j_{k1})} & \textbf{T}^{(j_{k2})} & \cdots & \textbf{T}^{(j_{k\eta_k})}
    \end{bmatrix},
\end{equation}
which has dimensions $S \times \eta_k(D + 1)$. Using this matrix, we obtain the linear model
\begin{equation}
    \textbf{w}_k' = \textbf{T}_k \textbf{q}_k + \textbf{z}_k'
    \label{eq_recov_model}
\end{equation}
where $\textbf{w}_k'$ is simply the first $S$ elements of $\textbf{w}_k$ in (\ref{delay_model}), and likewise with $\textbf{z}_k'$. In the absence of delay information, $\textbf{q}_k$ can be thought of as an unknown \textit{selection} vector. As was previously noted, each column of $\textbf{T}_k$ represents a unique discrete delay possibility for one of the $\eta_k$ users on that RE. As $\textbf{w}_k'$ is only a noisy sum of each user's delayed pilot sequence $\textbf{s}^{(j_{ki})}$, and assuming user $j_{ki}$'s pilot sequence is transmitted with power $\mathcal{E}^{(j_{ki})}$, the elements of $\textbf{q}_k$ will take a value of $\sqrt{\mathcal{E}^{(j_{ki})}}h_k^{(j_{ki})}$ in the rows corresponding to the true delays of the users (in other words, the ``correct'' columns of $\textbf{T}_k$) and zeros elsewhere. With knowledge of the pilot sequences, traditional linear estimation techniques can be used to recover $\textbf{q}_k$, and in turn estimate the delays $d_j$.

However, we can see that the dimensions of $\textbf{T}_k$ are $S\times \eta_k(D + 1)$. With the application of power-constrained devices in mind, we desire a short $\textbf{s}'^{(j)}$, and thus we assume that $S < \eta_k(D+1)$. Under this assumption, the system of linear equations given by (\ref{eq_recov_model}) is \textit{underdetermined}.  While this would normally be cause for concern, we can also see that of the $\eta_k(D + 1)$ elements in $\textbf{q}_k$, only $\eta_k$ will be nonzero. Thus, $\textbf{q}_k$ is potentially a highly sparse vector for which compressed sensing techniques can be employed for reliable recovery \cite{cs_donoho}. High-probability recovery is possible in this scenario by solving the LASSO optimization problem \cite{lasso}:
\begin{equation}
    \hat{\textbf{q}}_k = \argmin_{\textbf{q}_k \in \mathbb{C}^{\eta_k(D+1)}} \left\{ \frac{1}{2} \Vert \textbf{T}_k \textbf{q}_k - \textbf{w}_k \Vert^2_2 + \lambda \sigma \Vert \textbf{q}_k \Vert_1 \right\}
    \label{lasso}
\end{equation}
where $\lambda$ is a regularization parameter. The LASSO problem entails the optimization of a sum of convex functions, which is itself a convex function, and thus it can be efficiently solved with well-developed algorithms \cite{boyd}.

Once $\hat{\textbf{q}}_k$ is obtained by solving $(\ref{lasso})$, the delay information must be extracted. This is easily accomplished by locating the maximum-magnitude element in $\hat{\textbf{q}}_k$ corresponding to each user. For example, say that $D = 10$, and thus the first 11 elements of $\hat{\textbf{q}}_k$ correspond to the 11 columns of $\textbf{T}_k^{j_{k1}}$, i.e. the delay possibilities of the first user on that RE. Then the element with the largest magnitude of these 11 is decided as corresponding to the correct delay. Note that, as each user is associated with more than one RE, performing this estimation for each RE will inherently lead to multiple delay estimates for the same user. This is addressed by averaging the elements of $\hat{\textbf{q}}_k$ corresponding to a given user across all RE's, and then using this averaged vector to determine the delay. This also illustrates the natural benefit of diversity that can be exploited due to the structure of the SCMA system.

Once the delay estimates have been obtained, each sample can be associated to the correct symbol for each user, and MPA (or any other decoding algorithm) can be used to jointly decode the bits of each user. The full delay-estimation and decoding algorithm is given below.

\begin{algorithm}
    \caption{SCMA Uplink Delay Estimation and Decoding}
    \begin{algorithmic}[1]
    \label{main_alg}
    \REQUIRE $\{\textbf{w}_1,\ldots,\textbf{w}_K\}$, $\sigma^2$, $\{\textbf{h}_1,\ldots,\textbf{h}_K\}, \{\textbf{s}'^{(1)}, \ldots, \textbf{s}'^{(J)}\}$
    \FOR {$k$ in $K$}
        \STATE $\textbf{w}_k' = \textbf{w}_k[1:S]$
        \STATE obtain $\textbf{T}_k$ using (\ref{eq_Tk})
        \STATE $\hat{\textbf{q}}_k = \argmin_{\textbf{q}_k \in \mathbb{C}^{\eta_k(D+1)}} \left\{ \frac{1}{2} \Vert \textbf{T}_k \textbf{q}_k - \textbf{w}'_k \Vert^2_2 + \lambda \sigma \Vert \textbf{q}_k \Vert_1 \right\}$
    \ENDFOR
    \STATE use $\hat{\textbf{q}}_k$'s to estimate $\{d_1, \ldots, d_J\}$
    \STATE implement decoder using $\{d_j\}_{j=1}^J$, $\{\textbf{w}_k\}_{k=1}^K$ $\{\textbf{h}_k\}_{k=1}^K$, and $\sigma^2$
    \end{algorithmic}
\end{algorithm}

\subsection{Delay Estimation Theoretical Results}

To solve the LASSO problem given by (\ref{lasso}), we use the constant-step forward-backward algorithm, presented as Algorithm 3.4 in \cite{combettes}:
\begin{algorithm}
    \caption{Constant-Step Forward-Backward Algorithm}
    \begin{algorithmic}[1]
    \label{fb_alg}
    \STATE \textit{Initialization:} Fix $\varepsilon \in (0, 3/4)$ and $\textbf{x}_0 \in \mathbb{R}^{C\eta_k}$
    \WHILE {$\Vert \textbf{x}_n - \textbf{x}_{n-1} \Vert_2 > \epsilon$}
        \STATE $\textbf{y}_n = \textbf{x}_n - \beta^{-1}\nabla \frac{1}{2} \Vert \textbf{T}_k \textbf{x}_n - \textbf{w}_k \Vert^2_2$
        \STATE $\mu_n \in [\varepsilon, 3/2-\varepsilon]$
        \STATE $\textbf{x}_{n+1} = \textbf{x}_n + \mu_n(\prox_{\beta^{-1}\lambda \sigma \Vert \cdot \Vert_1} \textbf{y}_n - \textbf{x}_n)$
    \ENDWHILE
    \RETURN $\hat{\textbf{q}}_k = \textbf{x}_n$
    \end{algorithmic}
\end{algorithm}

Here, $\mu_n$ is a step size parameter, $\beta = \Vert \textbf{T}_k \Vert_2^2$ is the largest eigenvalue of $\textbf{T}_k$, and $\prox_{ \alpha f} \textbf{x}$ is the proximity operator of the function $f(\cdot)$ evaluated at $\textbf{x}$ with scalar parameter $\alpha > 0$. Algorithm \ref{fb_alg} belongs to a class of proximal splitting algorithms, and has been proven to converge to the solution of (\ref{lasso}) \cite{combettes}. 

We will now work to obtain some bounds on the optimality of the solution $\hat{\textbf{q}}_k$ for the specific problem at hand. For the remainder of this section, we drop the subscript $k$ as the process is identical for each RE. First, recall from (\ref{eq_Tk}) that $\textbf{T}$ is formed by concatenating $\eta$ Toeplitz matrices, and therefore is itself a Toelplitz matrix. Following the results in Haupt \textit{et al} \cite{Haupt:ToeplitzCS:10}, we show that this matrix satisfies the restricted isometry property (RIP) with high probability by showing that the eigenvalues of the Gram matrix $\textbf{T}^T \textbf{T}$ are close to 1. For a true isometry, these eigenvalues are equal to 1. As is well-known in the field of compressed sensing,   satisfaction of the RIP leads to optimal recovery using the LASSO technique \cite{wainwright_sparsity}.

The sketch of the proof is as follows. We show that the diagonal elements of the Gram matrix $\textbf{G} = \textbf{T}^T \textbf{T}$ are close to 1 using the Laurent-Massart bounds and that the off-diagonal elements are close to 0 with high probability using the Chernoff bound. Applying Gershogorin's Theorem, this provides a statement about the eigenvalues of $\textbf{G}$ and therefore the RIP.

The diagonal elements of $\textbf{G}$ are inner products of a user's pilot sequence $\textbf{s}^{(j)}$ with itself, while the off-diagonal elements are inner products of shifted (i.e., misaligned) versions of these sequences or inner products of different users' (possibly misaligned) pilot sequences. Essentially, each $\textbf{s}^{(j)}$ should be most strongly correlated with itself, and there should be weaker relationships between $\textbf{s}^{(j)}$ and its delay-shifted copies or with $\textbf{s}^{(i)}$ for some other user $i \in \Omega, i \neq j$. With this in mind, throughout the proof we assume that pilot symbols are generated as iid complex normal random variables. We first present the following lemmas which will be used in our main theorem:

\begin{lemma}[Gershgorin's Circle Theorem]
The eigenvalues of an $n\times n$ matrix M all lie in the union of discs $d_i(c_i, r_i)$ with center $c_i = M_{ii}$ and radius $\sum_{j=1, j\neq i}^n |M_{ij}|$.
\label{lem_gershgorin}
\end{lemma}

\begin{lemma}[Lemma 9, \cite{Haupt:ToeplitzCS:10}]
Let $x_i$ and $y_i$, $i = 1, \dots, P$, be sequences of i.i.d. zero-mean Gaussian random variables with variance $\sigma^2.$
Then 
\[ \Pr\left(\left|\sum_{i=1}^P x_iy_i \right| \geq t\right) \leq  2\exp\left(-\frac{t^2}{4\sigma^2(P\sigma^2+t/2)}\right)\]
\label{haupt}
\end{lemma}

\begin{lemma}
    Let $z_i$, $i=1, \dots, P$ be i.i.d. zero mean complex normal variables, i.e. $z_i = x_i +j y_i$ such that $x_i \indep y_i$, $x_i, y_i \sim \mathcal{N}(0, \sigma^2)$. Then
    \[ \Pr\left(\left|\sum_i^{P} |z_i|^2 - 2P\sigma^2\right| \geq 4\sigma^2\sqrt{2Pt} \right) \leq 2\exp(-t).\]
    \label{lem_same_prod}
\end{lemma}

\begin{proof}
    This follows from the Laurent-Massart \cite{Laurent:Adaptive:2000} bounds given in \cite{Haupt:ToeplitzCS:10}, where we note that $\sum_i |z_i|^2 = \sum_i x_i^2 + y_i^2.$ Because $x_i$ and $y_i$ are i.i.d., treat this as a sum of $2P$ i.i.d. squared normal random variables. 
    This gives the one-sided bounds
        \[ \Pr\left(\sum_i |z_i|^2 - 2P\sigma^2 \geq 2\sigma^2\sqrt{2Pt} + 2\sigma^2 t\right) \leq \exp(-t)\]
    and 
    \[ \Pr\left(2P\sigma^2 - \sum_i |z_i|^2 \geq 2\sigma^2 \sqrt{2Pt}\right) \leq \exp(-t).\]
    The result follows.

\end{proof}

\begin{lemma}
    Let $z_i$ and $w_i$, $i = 1, \dots, P$ be sequences of i.i.d. zero mean complex Gaussian random variables.
    
    Then 
    \[ \Pr\left(\left|\sum_{i=1}^P \bar{z_i}{w_i} \right| \geq t\right) \leq 4 \exp\left(-\frac{t^2}{16\sigma^2(2P\sigma^2 + t/4)}\right).\]
    \label{lem_dif_prod}
\end{lemma}

\begin{proof}
    Let $z_i = x_i + jy_i$ and $w_i = u_i + jv_i$. Then $\bar{z_i}w_i = (x_iu_i + y_iv_i) + j(x_iv_i-y_iu_i)$, and by the triangle inequality
    \[ |\sum_i \bar{z_i}{w_i}| \leq |\sum_i (x_iu_i + y_i v_i)| + |\sum_i (x_iv_i - y_iu_i)|.\]
    To illustrate, suppose $A, B, C \geq 0$ such that $A \leq B + C$. For any $\epsilon \geq 0$, 
    \[ \Pr (A \geq \epsilon) \leq \Pr (B + C \geq \epsilon). \]
    We also see that
    \[ \Pr (B+C \geq \epsilon) \leq \Pr(B \geq \epsilon/2 \textrm{ or } C \geq \epsilon/2), \]
    so we must have 
    \[ \Pr(A \geq \epsilon) \leq 2\max \{ \Pr(B \geq \epsilon/2), \Pr(C \geq \epsilon/2). \} \]
    Proceeding as above with $x_i, y_i, u_i, v_i$ i.i.d., we obtain
    \[ \Pr\left(\left|\sum_{i=1}^P \bar{z_i}{w_i} \right| \geq t\right) \leq 2 \Pr \left( \left|\sum_i^P x_iu_i + y_i v_i\right| \geq t/2\right), \]
    and the result follows from Lemma \ref{haupt}.
\end{proof}

With these results in place, our main results are presented in Theorem 1.

\begin{theorem}
Let the $P = S - D$ nonzero pilot symbols $\{s_{i}^{(j)}\}_{i=1}^P$, for $j \in \Omega$, be i.i.d. zero mean complex Gaussian random variables with variance $\sigma_p^2 = \frac{1}{2P}$ and assume $D \geq 1$. For any $\delta_S \in (0, 1)$ and sparsity level $\tilde{S} \leq c_2 \sqrt{P/\log(4\eta(\eta+1)(P-1) + 2\eta^2)}$, the matrix $\textbf{T} = [\textbf{T}^{(j_{1})} \cdots \textbf{T}^{(j_{\eta})}]$ satisfies the $(\tilde{S}, \delta_S)$-RIP with probability at least $1-\exp(-c_1 P/ (\tilde{S}-1)^2)$, where $c_1$ and $c_2$ depend only on $\delta_S$.
\label{thm_main}
\end{theorem}

\begin{proof}
We first would like to bound the probability that diagonal elements of the Gram matrix are far from 1. Denote $g_{ii}^{(j)} = \langle \textbf{s}^{(j)}, \textbf{s}^{(j)} \rangle$ where $\textbf{s}^{(j)}$ is the full pilot sequence for user $j$. Each of the $D+1$ diagonal elements $g_{ii}^{(j)}$ of a given diagonal block $(\textbf{T}^{(j)})^T \textbf{T}^{(j)}$ are identical.
Apply Lemma 3 with $z_i = s_{i}^{(j)}$ to get 
\[ \Pr(|g_{ii}^{(j)}-1| \geq \epsilon_d) \leq 2 \exp\left(\frac{-P\epsilon_d^2}{8}\right). \]
Taking the union over the $\eta$ diagonal blocks gives
\[ \Pr(\cup_{i=1}^{\eta}|g_{ii}^{(j)}-1| \geq \epsilon_d) \leq 2\eta\exp\left(\frac{-P\epsilon_d^2}{8}\right). \]

Next, we would like to bound the probability that the off-diagonal elements are far from 0. The off-diagonal elements of the diagonal blocks are inner products of shifted versions of the same pilot sequence. For example, 
\[g_{i+1, i}^{(j)} = \sum_{n=1}^{P-1} \overline{s_{n+1}^{(j)}} s_n^{(j)} = \overline{s_2^{(j)}}s_1^{(j)} + \cdots + \overline{s_{P}^{(j)}}s_{P-1}^{(j)}.\]
As in \cite{Haupt:ToeplitzCS:10}, we split this sum to separate dependent terms:
\[ \sum_{n=1}^{t_1} \overline{s_{\pi_1(n+1)}^{(j)}}s_{\pi_1(n)}^{(j)} + \sum_{n=1}^{t_2} \overline{s_{\pi_2(n+1)}^{(j)}}s_{\pi_2(n)}^{(j)},\]
where $\pi_1$ and $\pi_2$ are permutations to reorder the terms of the sum and $t_1$, $t_2$ are the corresponding number of elements (no greater than $\frac{P}{2}$.
Labeling these components $g^{(j)_1}$ and $g^{(j)_2}$, we see that 
\begin{align*}
     \Pr\left(|g_{i\ell}^{(j)} | \geq \frac{\epsilon_o}{\tilde{S}-1} \right) \leq  2 \max & \left\{ \Pr\left(|g_{i\ell}^{(j)_1} | \geq \frac{\epsilon_o}{2(\tilde{S}-1)} \right), \right. \\
    & \quad \left. \Pr\left(|g_{i\ell}^{(j)_2}| \geq \frac{\epsilon_o}{2(\tilde{S}-1)} \right)  \right\}
\end{align*}
Applying Lemma \ref{lem_dif_prod} to each of the terms on the right hand side, we have upper bounds with exponential arguments of the form
\[ -\frac{P^2\epsilon_o^2 }{4(\tilde{S}-1)^2(4t + \frac{\epsilon_o}{(\tilde{S}-1)}P)}. \]
Relying on the fact that $t \leq P$ and $\epsilon_o \leq \tilde{S}$, 
we reduce this to $- P \epsilon_o^2/20(\tilde{S}-1)^2$. 
Finally we take the union bound first over the distinct elements of a block then over the diagonal blocks producing a factor of $\eta (P-1)$: 
\begin{align*}
    &\Pr(\cup_{j\in \eta} \cup_{\ell=1}^{P-1} \{g_{i\ell}^{(j)} \geq  \frac{\epsilon_o}{\tilde{S}-1} \} )  \\
    &\quad \quad\quad\quad \quad \quad  \leq 8\eta (P-1) \exp\left(\frac{-P\epsilon_o^2}{20(\tilde{S}-1)^2}\right)
\end{align*} 

For the off-diagonal blocks, there are no issues of dependence, so we may apply the Lemma \ref{lem_dif_prod} directly. By carefully counting the number of nonzero and non-identical entries in the off-diagonal blocks, the total number of distinct terms is found to be $(\eta^2 - \eta)(2P-1)/2$. Thus, taking the union bound yields
\begin{align*} 
&\Pr(\cup_{O = 1}^{\eta^2-\eta} \cup_{i, \ell = 1}^{D^2 + D} \{g_{i\ell}^{(j)} \geq \frac{\epsilon_o}{\tilde{S}-1} \} ) \\
& \quad \quad \quad \leq 2(\eta^2-\eta)(2P-1)\exp\left(-\frac{P \epsilon_o^2}{10(\tilde{S}-1)^2}\right).
\end{align*}

As a consequence of Lemma \ref{lem_gershgorin}, for $\textbf{T}$ to satisfy the RIP we require that $\epsilon_d + \epsilon_o = \delta_s$. Thus, we let $\epsilon_d = \delta_S/3$ and $\epsilon_o = 2\delta_S/3$. Combining the above bounds, we see that
\begin{align*}
     \Pr(\textbf{T} & \text{ does not satisfy } RIP) \leq \\
    & \left( 4\eta(\eta+1)(P-1) + 2\eta^2 \right) \exp\left(\frac{-P \delta_S^2}{45 (\tilde{S}-1)^2}\right).
\end{align*}
The result follows with
$$c_1 < \delta_S^2 / 45, \quad c_2 = \sqrt{(\delta_S^2 - 45c_1)/45}.$$

\end{proof}

Theorem \ref{thm_main} can be used to determine the required number of nonzero pilot symbols for guaranteed delay recovery. By choosing a desired probability, a lower bound on $P$ can be obtained. This bound can be lowered by decreasing the variance of the iid pilot symbols, which was assumed to be $1/2P$ throughout Theorem \ref{thm_main}. Note however, that Theorem \ref{thm_main} does not consider the effect of noise. While decreased pilot power will yield improved recovery in low-noise conditions, this may have adverse affects in high-noise environments. This introduces an application-specific design trade-off that must be considered in a practical implementation. We plan to investigate the implications of noise in our future work.

\section{Simulation Results}
\label{sim_section}

To reduce the complexity of MUD incurred by MPA, we employ the parallel Monte Carlo-Markov Chain (MCMC) SCMA decoder, first introduced by Chen \textit{et al.} \cite{chen}. While the complexity of nearly all SCMA decoders increases exponentially with the number of users sharing a resource element \cite{scma_survey}, the complexity of the parallel MCMC decoder scales linearly with this parameter, unlocking the possibility of a massive implementation. The parallel MCMC decoding algorithm from \cite{chen} is given here in Algorithm \ref{mcmc_alg} for completeness.

\begin{algorithm}
    \caption{Parallel MCMC SCMA Decoder}
    \begin{algorithmic}[1]
    \label{mcmc_alg}
    \STATE \textit{Initialization:} Randomly select $\textbf{x}^{j,(0,n_2)}$ from $\mathcal{C}_j$ $(j=1,\ldots,J)$
    \STATE Set $\gamma_{b_i^k = \pm1}^{\textrm{max}} = - \infty$
    \FOR {$n_1=1$ to $N_s$}
        \FOR {$n_2=1$ to $N_2$}
            \FOR {$j = 1$ to $J$}
                \STATE Draw samples $\textbf{x}^{j,(n_1,n_2)}$
                \STATE Demap $\textbf{b}^{j,(n_1,n_2)}$ by $\textbf{b}^j = (f^j)^{-1}(x_1^j, x_2^j, \ldots, x_K^j)$
                \STATE Let $b_i^{j,(n_1,n_2)} = \pm 1$
                \STATE $\textbf{x}_{b_i^k = \pm 1}^{j,(n_1,n_2)} = f^j(b_1^j, b_2^j, \ldots, b_{N_B}^k \vert b_i^{j,(n_1,n_2)} = \pm 1)$
            \ENDFOR
            \STATE $\gamma_{b_i^j = \pm 1}^{(n_1,n_2)} = -\frac{1}{\sigma_z^2} \left\Vert \textbf{y} - \sum_{m\neq j}^J \bar{\textbf{x}}^{m,(n_1,n_2)} - \textbf{x}_{b_i = \pm 1}^{k,(n_1,n_2)} \right\Vert^2$
            \STATE let $\gamma_{b_i^j = \pm 1}^{\textrm{max}} = \gamma_{b_i^j = \pm 1}^{(n_1,n_2)}$ if $\gamma_{b_i^j = \pm 1}^{\textrm{max}} < \gamma_{b_i^j = \pm 1}^{(n_1,n_2)}$ $i = 1,\ldots,N_B$
        \ENDFOR
    \ENDFOR
    \RETURN $\lambda(b_i^j) = \gamma_{b_i^j = + 1}^{\textrm{max}} - \gamma_{b_i^j = - 1}^{\textrm{max}}$
    \end{algorithmic}
\end{algorithm}

To test the performance of the proposed scheme, we simulate an SCMA system operating under the frame structure specified by the current 5G standards. That is, we consider the case in which 12 OFDM tones, with a subcarrier spacing of 30kHz, are treated as orthogonal resource elements, and users transmit packets the length of a single frame, or 10ms period. Thus, decoding takes place over a single resource block. In practice, this could represent a system in which many devices periodically transmit very small quantities of data, e.g. a sensor network utilizing mMTC. Each user/device utilizes two RE's, and three users share each RE, resulting in a total of 18 users communicating over a single resource block. We employ the low-complexity SCMA codebook design scheme outlined in \cite{yang_codebook_design}, which was designed for large-scale SCMA. 1/4 polar encoding is applied to the bits of each user before transmission. Under this setup, each user transmits $N = 224$ channel-coded QAM symbols over each frame, or 224/4 = 56 data symbols. The pilot length is chosen to be $S = 56$, and the maximum delay is assumed as $D = 42$, or 1.5ms. Pilot sequences are generated as sequences of iid complex normal random variables with power $\mathcal{E}^{(j)}$.

First, we demonstrate the performance of the proposed delay estimation technique for the AWGN channel. We simulate asynchronous transmissions from each user according to the setup described above, perform delay recovery over the pilot sequences, and then compute the mean absolute error (MAE) of the recovered delays. Results for various pilot sequence powers $\mathcal{E}^{(j)}$ obtained using the proposed LASSO method are given in Figure \ref{delay_est_results_lasso}. Results using the well-known least squares estimation technique are given in Figure \ref{delay_est_results_LS} for comparison. As expected, the LASSO method greatly outperforms the least squares method, due to the underdetermined nature of the problem.

\begin{figure}[htbp]
    \centerline{\includegraphics[scale = 0.55]{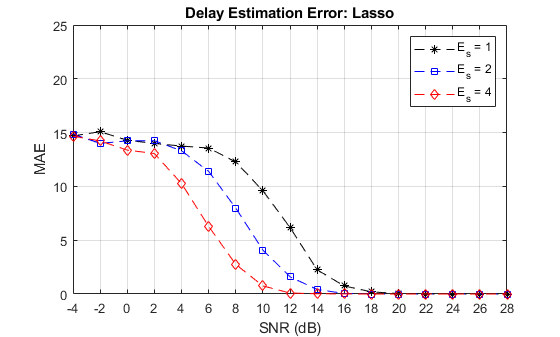}}
    \caption{Mean absolute error of LASSO delay estimation}
    \label{delay_est_results_lasso}
\end{figure}

\begin{figure}[htbp]
    \centerline{\includegraphics[scale = 0.55]{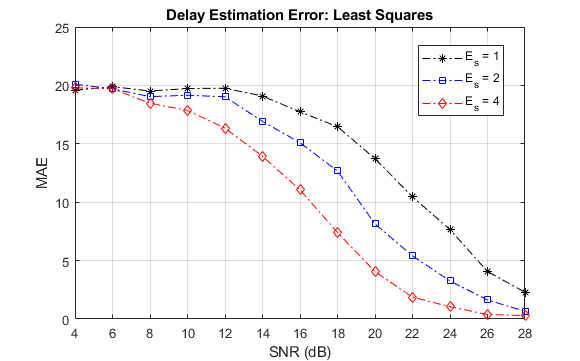}}
    \caption{Mean absolute error of least squares delay estimation}
    \label{delay_est_results_LS}
\end{figure}

Next we simulate the entire asynchronous end-to-end system for the AWGN channel. The Parallel MCMC decoder is implemented using 15 sampling iterations that are carried out over four parallel sampling chains, with mixing parameter $\mu = 10$. For these tests, the power of the pilot sequences for all users is held constant at $\mathcal{E}^{(j)} = 1$. The bit error rate (BER) was recorded for several SNR values, and the results are presented in Figure \ref{awgn_ber_fig}. Here, the asynchronous system with the proposed delay estimation is compared to a synchronous system employing only the Parallel MCMC decoder with no delay estimation.

\begin{figure}[htbp]
    \centerline{\includegraphics[scale = 0.60]{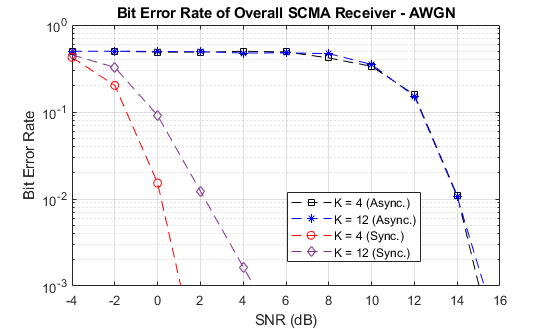}}
    \caption{Overall BER simulation results - AWGN channel}
    \label{awgn_ber_fig}
\end{figure}

\begin{figure}[htbp]
    \centerline{\includegraphics[scale = 0.50]{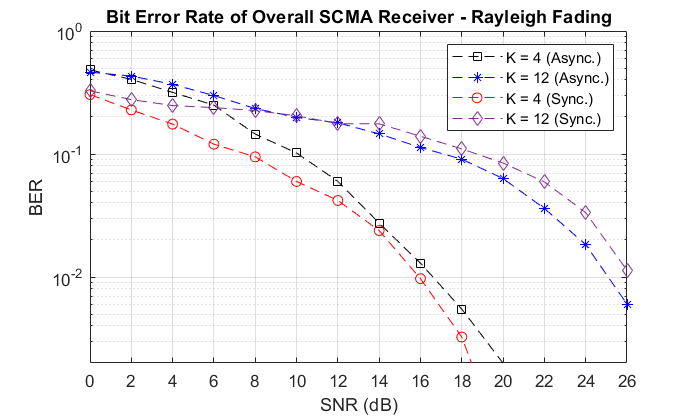}}
    \caption{Overall BER simulation results - Rayleigh fading channel}
    \label{rayleigh_ber_fig}
\end{figure}

From Figure \ref{awgn_ber_fig}, we see the impact of unknown delays on the asynchronous scheme. Also, it is interesting to note that both asynchronous systems (4 and 12 RE) appear to demonstrate identical performance, though we would expect the 4 RE system to outperform the 12 RE system, as in the synchronous trials. To make sense of this, turn again to Figure \ref{delay_est_results_lasso}, which shows that at $\mathcal{E}^{(j)} = 1$, the delay estimation begins to improve around 6dB; this corresponds exactly to where the performance of the asynchronous systems in Figure \ref{awgn_ber_fig} begin to improve as well. From this observation, we can see that the performance of the overall asynchronous system is limited by the delay estimation step, which we can improve by increasing the power of the pilot sequence, as is shown in Figure \ref{delay_est_results_lasso}.

Next, we repeat these experiments for the Rayleigh fading channel. In these experiments, the power of the pilot sequences is increased to $\mathcal{E}^{(j)} = 5$ to compensate for the effects of fading. From Figure \ref{rayleigh_ber_fig}, we see that the both asynchronous and synchronous systems are greatly impacted by fading, particularly as the size of the system grows. For the $K = 4$ scenario, we see that performance of the asynchronous receiver begins to converge to that of the synchronous receiver around 14dB, though it never fully converges. This is due to the fact that there is always a positive probability for a deep fade, which will severely impact delay recovery despite a high SNR. Furthermore, we see that the asynchronous receiver actually \textit{outperforms} the synchronous receiver for high SNR in the $K = 12$ scenario, beginning around 14dB. This is likely due to the fact that, in the asynchronous scheme, the first received bits experience no interference from other users (see Figure \ref{async_scma_fig}). Therefore, if delay information is known, it actually becomes easier to decode these bits, leading to improved performance.

\section{Conclusion}
\label{con_sec}

In this paper, we present a novel delay estimation scheme for an asynchronous uplink SCMA system. First, the asynchronous system model is defined, in which each user's transmitted signal experiences some unknown delay at the receiver. Then, we formulate the delay estimation problem as a sparse signal recovery problem, which can then be solved with standard optimization techniques. Furthermore, we provide bounds on delay recovery under the assumption of iid Gaussian pilot symbols, which can be used to inform system design. Finally, the performance of the proposed scheme is illustrated through various simulations. It is shown that the proposed scheme is able to accurately estimate delays and a receiver implementing the proposed scheme can achieve high reliability in the presence of AWGN and fading channels.

The results in this paper pave the way for a more robust receiver at the uplink of an SCMA system. There are many opportunities to expand upon this work. One avenue is to consider the joint design of the delay estimation scheme and decoder. The delay estimation scheme proposed in this work is decoder-agnostic, and thus can be used with any existing synchronous SCMA decoder. A joint design may yield improved performance. Taking this a step further, the SCMA codebook design, delay-estimation scheme, and decoder can all be designed together in a joint manner. While this represents a difficult problem, recent advances in deep learning and model-based deep learning provide a promising approach. Another opportunity for future work lies in the delay estimation scheme itself. While the traditional LASSO problem is used to recovery the sparse signal in this paper, other approaches (both convex and non-convex) using different penalty functions can sometimes yield better performance. It would be worthwhile to explore these other approaches to solving the sparse-recovery delay estmation problem, as well as the general channel estimation problem in which the complex channel coefficients are recovered as well.

% \begin{figure*}[htbp]
%     \centerline{\includegraphics[]{}}
%     \caption{}
%     \label{}
% \end{figure*}

\printbibliography

\begin{IEEEbiography}[{\includegraphics[width=1in,height=1.25in,clip,keepaspectratio]{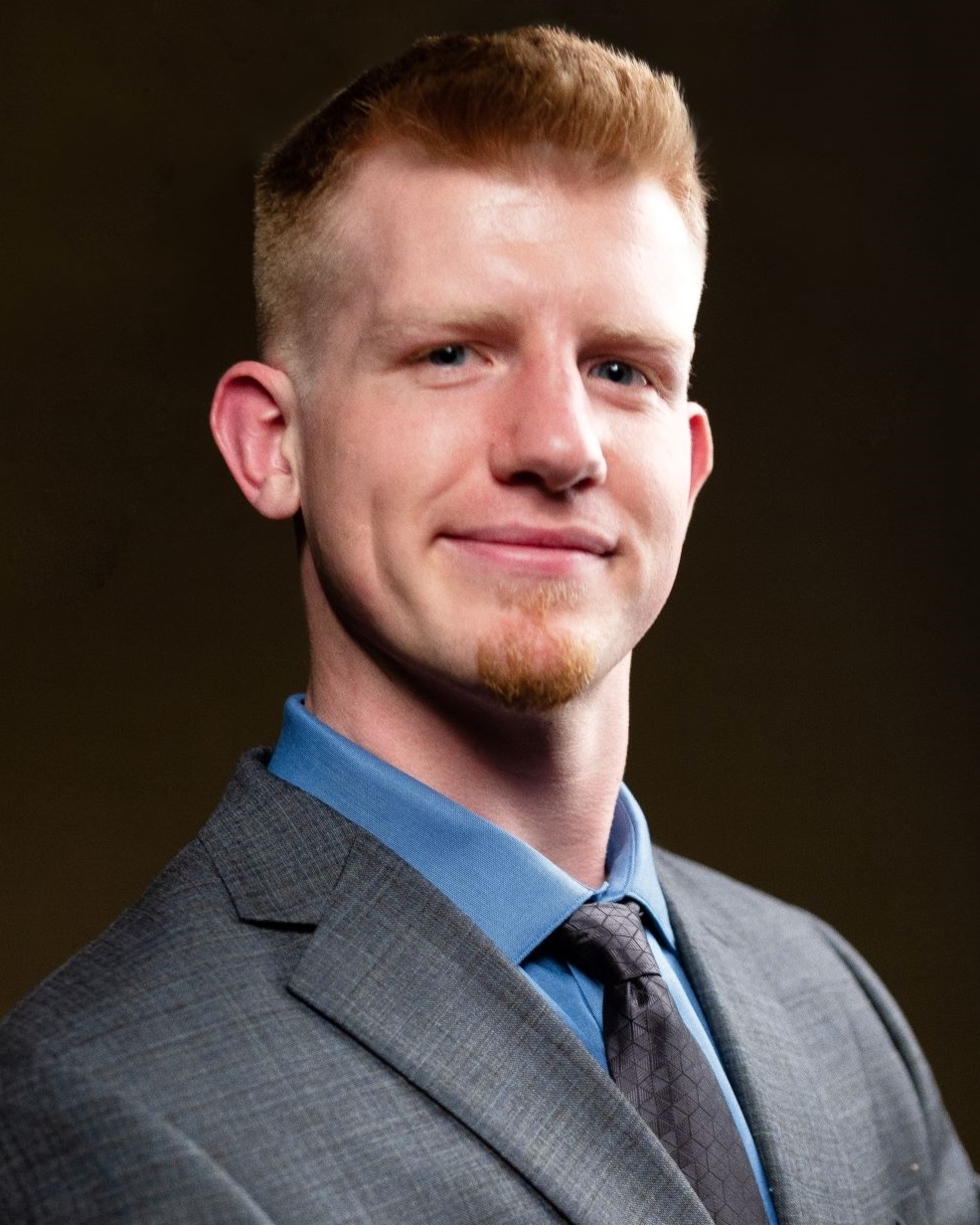}}]%
{Dylan Wheeler}
(Graduate Student Member, IEEE) received the A.S. degree from Highland Community College, Highland, KS, USA in 2016, the B.S. degree in Engineering from Ottawa University, Ottawa, KS, USA in 2018, and the M.S. degree in Electrical and Computer Engineering from Kansas State University, Manhattan, KS, USA in 2021. He is currently a Ph.D. student and a member of the Cyber-Physical Systems and Wireless Innovations Research Group at Kansas State University, Manhattan, KS, USA. His research interests include semantic communications, machine learning and artificial intelligence, and internet-of-things technologies for beyond-5G wireless networks.
\end{IEEEbiography}

\begin{IEEEbiography}[{\includegraphics[width=1in,height=1.25in,clip,keepaspectratio]{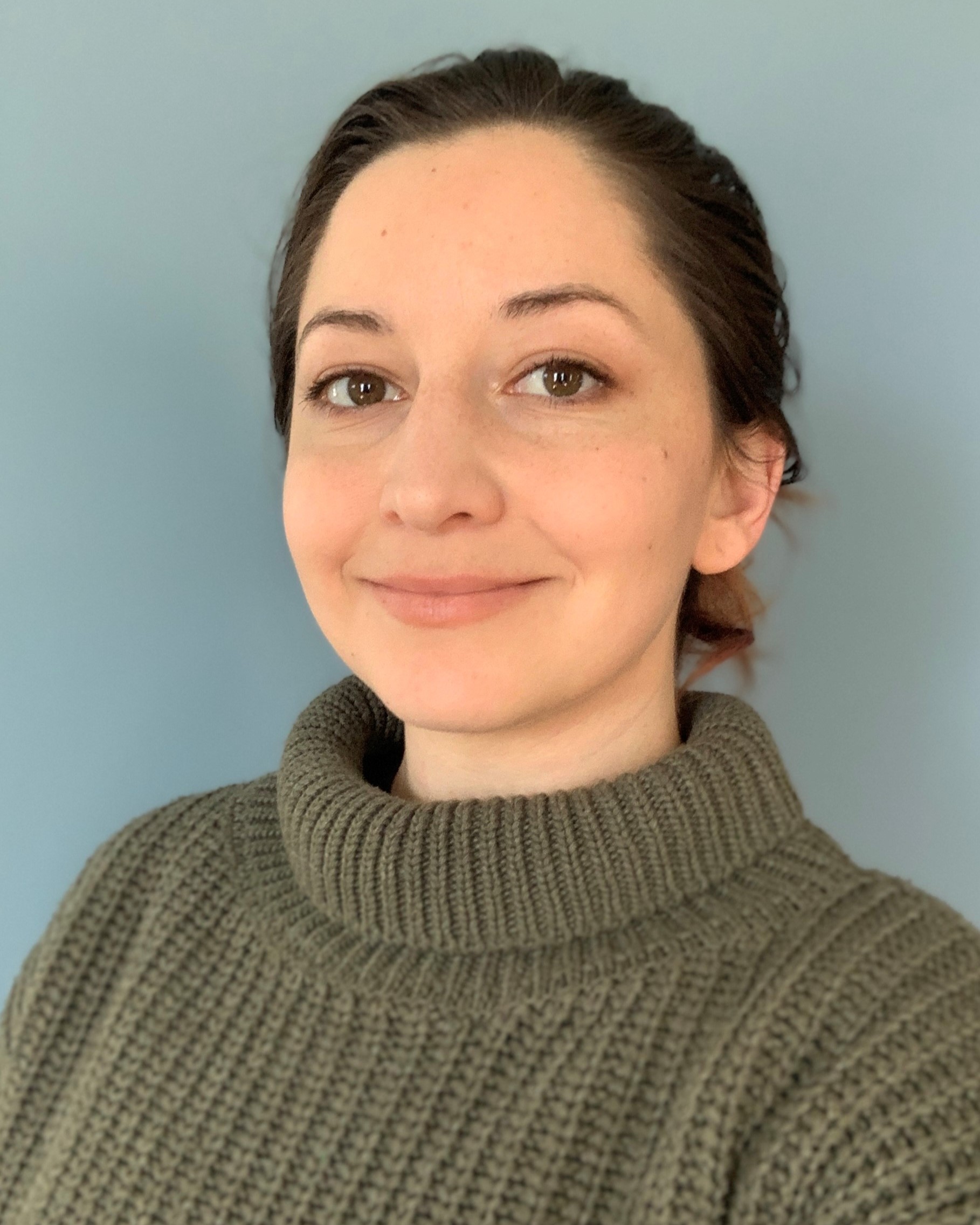}}]%
{Erin E. Tripp}
received the B.S. degree in mathematics from the University of California, Santa Barbara, Santa Barbara, CA, USA, in 2013, the M.S. and Ph.D. degrees in mathematics from Syracuse University, Syracuse, NY, USA, in 2017 and 2019, respectively.
She is currently a Research Mathematician with the Air Force Research Laboratory Information Di-rectorate, Rome, NY, USA, working in optimization theory with applications to signal and image processing and machine learning.
\end{IEEEbiography}

\vskip 0pt plus -1fil

\begin{IEEEbiography}[{\includegraphics[width=1in,height=1.25in,clip,keepaspectratio]{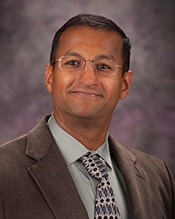}}]%
{Balasubramaniam Natarajan}
(Senior Member, IEEE) received the B.E. degree (Hons.) in electrical and electronics engineering from Birla Institute of Technology and Science, Pilani, India, Ph.D. degree in electrical engineering from Colorado State University, Fort Collins, CO, USA, Ph.D. degree in Statistics from Kansas State University, Manhattan, KS, USA, in 1997, 2002, and 2018, respectively. He is currently a Clair N. Palmer and Sara M. Palmer Endowed Professor and the Director of the Cyber-Physical Systems and Wireless Innovations Research Group. His research interests include statistical signal processing, stochastic modeling, optimization, and control theories. He has worked on and published extensively on modeling, analysis and networked estimation and control of smart distribution grids and cyber physical systems in general. He has published over 200 refereed journal and conference articles and has served on the editorial board of multiple IEEE journals including IEEE Transactions on Wireless Communications.
\end{IEEEbiography}

\end{document}